\pdfoutput=1
\documentclass[9pt]{IEEEtran}
\usepackage{cite}
\usepackage{dcpic,pictexwd}
\usepackage{graphicx}
\usepackage{url}
\usepackage{verbatim}
\usepackage{color}
\definecolor{lgray}{gray}{0.92}
\definecolor{lblue}{rgb}{0.90,0.90,1.00}
\definecolor{lyellow}{rgb}{1.00,1.00,0.70}

\usepackage{listings}
\newenvironment{codex}{\small\verbatim}{\endverbatim\normalsize}

\lstnewenvironment{code}
    {\lstset{}%
      \csname lst@SetFirstLabel\endcsname}
    {\csname lst@SaveFirstLabel\endcsname}
\newtheorem{prop}{Proposition}

\newtheorem{df}{Definition}

\newcommand{\BI}[0]{\begin{itemize}}
\newcommand{\EI}[0]{\end{itemize}}

\newcommand{\BE}[0]{\begin{enumerate}}
\newcommand{\EE}[0]{\end{enumerate}}

\newcommand{\BX}[0]{\begin{codex}}
\newcommand{\EX}[0]{\end{codex}}
\newcommand{\BV}[0]{\begin{verbatim}}
\newcommand{\EV}[0]{\end{verbatim}}

\def \bscale1 {0.25}
\def \bscale {0.25}
\def \N {\mathbb{N}}

\newcommand{\FIG}[4]{
\begin{figure}[htbp]
\centering
{\includegraphics[scale=#3]{figs/#4}}
\caption{#2}
\label{#1}
\end{figure}
}

\usepackage{amsfonts,amsmath}

\begin{document}

\title{
  On Two Infinite Families of Pairing Bijections
}

%\author{Paul~Tarau\thanks{Department of Computer Science and Engineering,University of North %Texas}}

\author{\IEEEauthorblockN{{\bf Paul Tarau}}\\
\IEEEauthorblockA{Department of Computer Science and Engineering\\
   University of North Texas\\ 
   Denton, Texas\\
   {\em Email: tarau@cs.unt.edu}\\
}}

\begin{comment}
\author{Paul Tarau}
\institute{
   {Department of Computer Science and Engineering}\\
   {University of North Texas}\\ 
   {Denton, Texas}\\
   {\em tarau@cs.unt.edu}\\
}
\end{comment}
\maketitle
\date{}

\begin{abstract}
We describe two general mechanisms
for producing pairing bijections
(bijective functions defined from $\N^2 \to \N$).
The first mechanism, using $n$-adic valuations
results in parameterized algorithms
generating a countable family of distinct pairing bijections.
The second mechanism, using characteristic functions of
subsets of $\N$ provides $2^\N$ distinct pairing bijections.
%The mechanisms are also extended
%to tupling bijections (defined from $\N^k \to \N$).
Mechanisms to combine such pairing
functions and their application to generate families of
permutations of $\N$ are also described.
The paper uses a small subset of the functional language
Haskell to provide type checked executable
specifications of all the functions defined in a
{\em literate programming} style. 
The self-contained Haskell code extracted from the paper
is available at 
\url{http://logic.cse.unt.edu/tarau/research/2012/infpair.hs} .

\end{abstract}

\begin{codeh}
module InfPair where
import Visuals
import Data.Bits
--import Pi
\end{codeh}
\section{Introduction}

\begin{comment}
a comparaison to constructive reals 

nAdicHead,nAdicTail

induced by a characteristic function

bootstrapped
\end{comment}

\begin{df}
A {\em pairing bijection} is a bijection $f:\N \times \N \to \N$. 
Its inverse $f^{-1}$ is called an {\em unpairing} bijection.
\end{df}

We are emphasizing here the fact
that these functions are bijections as the name {\em pairing
function} is sometime used in the literature to indicate
injective functions from $\N \times \N$ to $\N$.

Pairing bijections have been used in the first half of 19-th century by
Cauchy as a mechanism to express duble summations as
simple summations in series. They have been made famous
by their uses in the second half of the 19-th century
by Cantor's work on 
foundations of set theory. Their most
well known application is 
to show that infinite 
sets like $\N$ and $\N \times \N$ have the same cardinality.
A classic use in the theory of recursive functions is
to reduce functions on multiple arguments to single argument
functions. Reasons on why they are an interesting
object of study in terms of practical applications ranging
from multi-dimensional dynamic arrays to proximity
search using space filling curves are described in 
\cite{DBLP:conf/ipps/Rosenberg02a,lawder99,lawder:2000,faloutsos:2001}.

Like in the case of Cantor's original function $f(x,y)={1\over2}(x+y)(x+y+1)+y$,
pairing bijections have been usually hand-crafted by putting to
work geometric or arithmetic intuitions. 

While it is easy
to prove (non-constructively) that there is an uncountable
family of distinct pairing bijections, we have not seen in the
literature general mechanisms for
building families of 
pairing bijections indexed by $\N$ or $2^\N$.
It is even easier to generate (constructively) a countable family of
pairing functions simply by modifying its result of
a fixed pairing function with a reversible
operation (e.g XOR with a natural number, 
seen as the index of the family).

This paper introduces two general mechanisms for
generating such families, using $n$-adic valuations (section \ref{nadic})
and characteristic
functions of subsets of $\N$ (section \ref{char}), followed
by a discussion of related work (section \ref{rel}) and our 
conclusions (section \ref{concl}).

We will give here a glimpse of why our arguably more complex
pairing bijections are interesting.

The $n$-adic valuation based pairing functions will provide
a general mechanism for designing strongly asymmetric pairing
functions, where changes in one of the arguments have an
exponential impact on the result.

The characteristic-function mechanism, while intuitively obvious,
opens the doors, in combination with a framework providing
bijections between them and arbitrary data-types \cite{everything},
to custom-build arbitrarily intricate pairing functions associated
to for instance to ``interesting'' sequences of natural numbers
or binary expansions of \cite{intseq}
real numbers.

We will use a subset of the non-strict functional language Haskell
(seen as an equational notation for typed $\lambda$-calculus)
to provide executable definitions of mathematical functions on $\N$,
pairs in $\N \times \N$, subsets of $\N$, and sequences of natural numbers.
We mention, for the benefit of the
reader unfamiliar with
the language, that a notation like {\tt f x y} stands for $f(x,y)$,
{\tt [t]} represents sequences of type {\tt t} and a type declaration
like {\tt f :: s -> t -> u} stands for a function $f: s \times t \to u$
(modulo Haskell's ``currying'' operation, given the isomorphism between 
the function spaces ${s \times t} \to u$ and ${s \to t} \to u$). 
Our Haskell functions are always represented as sets
of recursive equations guided by pattern matching, conditional
to constraints (simple arithmetic relations following \verb~|~ and before
the \verb~=~ symbol).
Locally scoped helper functions are defined in Haskell
after the {\tt where} keyword, using the same equational style.
The composition of functions {\tt f} and {\tt g} is denoted {\tt f . g}.
It is also customary in Haskell, when defining functions in an equational style (using {\tt =})
to write $f=g$ instead of $f~x=g~x$ (``point-free'' notation).
The use of Haskell's ``call-by-need'' evaluation
allows us to work with infinite
sequences, like the {\tt [0..]} infinite list notation, corresponding to the
set $\N$ itself. 

\section{Deriving Pairing Bijections from $n$-adic valuations} \label{nadic}

We  first overview a mechanism for
deriving pairing bijections
from one-solution Diophantine equations.
Let us observe that
\begin{prop}
$\forall z \in \N^+=\N-\{0\}$ the Diophantine equation
\begin{equation}\label{dio}
2^x(2y+1)=z
\end{equation}
has exactly one solution $x,y \in \N$.
\end{prop}
This follows immediately from the unicity of the decomposition of a natural
number as a multiset of prime factors. 
Note that a slight modification of equation \ref{dio} results in the
pairing bijection originally introduced
in \cite{pepis,kalmar1},
seen as a mapping between the pair $(x,y)$ and $z$.
\begin{equation}\label{diopair}
2^x(2y+1)-1=z
\end{equation}

We will generalize this mechanism to obtain a family of
bijections between $\N \times \N$ and $\N^+$ (and
the corresponding pairing bijections between $\N \times \N$ and $\N$)
by choosing an arbitrary base $b$ instead of $2$.

\begin{df}
Given a number $n\in \N,~n>1$, the $n$-adic valuation of a 
natural number $m$ is the largest exponent $k$ of $n$,
such that $n^k$ divides m. It is denoted $\nu_n(m)$.
\end{df}
Note that the solution $x$ of the equation (\ref{dio}) is actually $\nu_2(z)$.
This suggest deriving similar Diophantine 
equations for an arbitrary $n$-adic valuation. 
We start by observing that the following holds:

\begin{prop}\label{qm}
$\forall b \in \N, b>1, \forall y \in \N$ 
if  $\exists q, m$ such that $b>m>0, y=bq+m$,
then
there's exactly one pair $(y', m')$, $b-1>m'\geq 0$ 
such that $y'=(b-1)q+m'$ and
the function associating $(y',m')$ to $(y,m)$ is a bijection.
\end{prop}
\begin{proof}
$y=bq+m, b>m>0$ can be rewritten as
$y-q-1=bq-q+m-1, b>m>0$, or equivalently
$y-q-1=(b-1)q+(m-1), b>m>0$ from where it
follows that setting
$y'=y-q-1$ and
$m'=m-1$
ensures the existence and unicity of y' and m' such that 
$y'=(b-1)q+m'$ and $b-1>m'>0$.
We can therefore define a function $f$ that transforms a pair $(y,m)$, 
such that $y=bq+m$ with $b>m>0$, into a pair $(y',m')$, 
such that $y'=q(b-1)+m'$ with $b-1>m' \geq 0$.
Note that the transformation works also in the opposite direction
with $y'=y-q-1$ giving $y=y'+q+1$, and with $m'=m-1$ 
giving $m=m'+1$. Therefore $f$ is a bijection.
\end{proof}

\begin{prop}
$\forall b \in \N,b>1, ~\forall z \in \N,z>0$ the 
system of Diophantine equations and inequations
\begin{equation}
b^x*(y'+q+1)=z
\end{equation}
\begin{equation}
y'=(b-1)q+m'
\end{equation}
\begin{equation}
b-1>m'\geq 0
\end{equation}
has exactly one solution $x,y' \in \N$.
\end{prop}
\begin{proof}
Let $f^{-1}$ be the inverse of the bijection $f$ defined in Proposition \ref{qm}.  Then $f^{-1}$
provides the desired unique mapping, that gives $y=y'+q+1$ and $m=m'-1$ such that $b>m>0$.
Therefore $y \equiv m~(mod~b)$ with $m>0$. And as $y$ is not divisible
with $b$, we can determine uniquely $x$ as the largest power of $b$ dividing $z$,
$x=\nu_b(z)$.
\end{proof}

We implement, for and arbitrary $b \in \N$, the Haskell code corresponding to these
bijections as the functions  {\tt nAdicCons b} and {\tt nAdicDeCons b}, 
defined between $\N \times \N$ and $N^{+}$. 
%{\small
\begin{code}
nAdicCons :: N->(N,N)->N
nAdicCons b (x,y')  | b>1 = (b^x)*y where
  q = y' `div` (b-1)
  y = y'+q+1
\end{code}
\begin{code}  
nAdicDeCons :: N->N->(N,N)
nAdicDeCons b z | b>1 && z>0 = (x,y') where
  hd n = if n `mod` b > 0 then 0 else 1+hd (n `div` b)
  x = hd z
  y = z `div` (b^x)
  q = y `div` b
  y' = y-q-1 
\end{code}
%}
Using {\tt nAdicDeCons} we define the head and tail projection functions {\tt nAdicHead} and {\tt nAdicTail}:
\begin{code}
nAdicHead, nAdicTail :: N->N->N
nAdicHead b = fst . nAdicDeCons b
nAdicTail b = snd . nAdicDeCons b
\end{code}

The following examples illustrate the operations for base {\tt 3}:
\begin{codex}
*InfPair> nAdicCons 3 (10,20)
1830519
*InfPair> nAdicHead 3 1830519
10
*InfPair> nAdicTail 3 1830519
20
\end{codex}
Note that {\tt nAdicHead n x} computes the $n$-adic valuation of x, $\nu_n(x)$ while
the tail corresponds to the ``information content'' extracted from the
remainder, after division by $\nu_n(x)$.

\begin{df}
We call
the natural number computed by {\tt nAdicHead n x} the $n$-adic head of $x \in \N^{+}$, 
by {\tt nAdicTail n x} the $n$-adic tail 
of $x \in \N^{+}$ and the natural number in $\N^+$ computed
by {\tt nAdicCons n (x,y)} the $n$-adic cons of 
$x,y \in \N$.
\end{df}

\begin{comment}
\begin{df}
A {\em pairing bijection} is a bijection $f:\N \times \N \rightarrow \N$. 
Its inverse $f^{-1}$ is called an {\em unpairing} bijection.
\end{df}
\end{comment}

By generalizing the mechanism shown for the equations \ref{dio} and \ref{diopair}
we derive from {\tt nAdicDeCons} and {\tt nAdicCons} 
the corresponding {\em pairing} and {\em unpairing} 
bijections {\tt nAdicPair} and {\tt nAdicUnPair}:
\begin{code}
nAdicUnPair :: N->N->(N,N)
nAdicUnPair b n = nAdicDeCons b (n+1)

nAdicPair :: N->(N,N)->N
nAdicPair b xy = (nAdicCons b xy)-1
\end{code}
One can see that we obtain a countable family of bijections  
$f_b: \N \times \N \rightarrow \N$ indexed by $b \in \N$,
$b>1$.

The following examples illustrate the work of these bijections for $b=3$. Note the
use of Haskell's higher-order function ``{\tt map}'', that applies the function
{\tt nAdicUnPair 3} to a list of elements and collects the results to a list, and the special value ``{\tt it}'',
standing for the previously computed result.
\begin{codex}
*InfPair> map (nAdicUnPair 3) [0..7]
[(0,0),(0,1),(1,0),(0,2),(0,3),(1,1),(0,4),(0,5)]
*InfPair> map (nAdicPair 3) it
[0,1,2,3,4,5,6,7]
\end{codex}
% map ((nAdicPair 3). (nAdicUnPair 3)) [0..15]
%\end{comment}

\subsubsection{Deriving bijections between $\N$ and $[\N]$}
For each base {\tt b>1}, we can also obtain a pair of bijections 
between natural numbers and lists of natural
numbers in terms of {\tt nAdicHead}, {\tt nAdicTail} and {\tt nAdicCons}:
\begin{code}
nat2nats :: N->N->[N]
nat2nats _ 0 = []
nat2nats b n | n>0 = 
   nAdicHead b n : nat2nats b (nAdicTail b n)
\end{code}
\begin{code}
nats2nat :: N->[N]->N
nats2nat _ [] = 0
nats2nat b (x:xs) = nAdicCons b (x,nats2nat b xs)
\end{code}

The following example illustrate how they work:
\begin{codex}
*InfPair> nat2nats 3 2012
[0,2,2,0,0,0,0]
*InfPair> nats2nat 3 it
2012
\end{codex}

Using the framework introduced in \cite{calc09fiso,everything} 
and summarized in the {\bf Appendix}, 
we can ``reify'' these bijections as {\tt Encoders} between natural numbers and sequences
of natural numbers (parameterized by the first argument of {\tt nAdicHead} and
{\tt nAdicTail}).
Such Encoders can now be ``morphed'', by using the bijections provided
by the framework, into various data types sharing the same
``information content'' (e.g. lists, sets, multisets).
\begin{code}
nAdicNat :: N->Encoder N
nAdicNat k = Iso (nat2nats k) (nats2nat k)
\end{code}

In particular, for $k=2$, we obtain the {\tt Encoder} corresponding to
the Diophantine equation (\ref{dio}) 
\begin{code}
nat :: Encoder N
nat = nAdicNat 2
\end{code}
The following examples illustrate these operations,
lifted through the framework defining 
bijections between datatypes, given in {\bf Appendix}.
\begin{codex}
*InfPair> as (nAdicNat 3) list [2,0,1,2]
873
*InfPair> as (nAdicNat 7) list [2,0,1,2]
27146
*InfPair> as nat list [2,0,1,2]
300
*InfPair> as list nat it
[2,0,1,2]
\end{codex}

\subsubsection{Deriving new families of Encoders and Permutations of $\N$}

For each $l,k \in \N$ one can
generate a family of permutations
(bijections $f:\N\rightarrow \N$), parameterized by the pair {\tt (l,k)}, by composing
{\tt nat2nats l} and {\tt nats2nat k}.
\begin{code}
nAdicBij :: N -> N -> N -> N
nAdicBij k l = (nats2nat l) . (nat2nats k) 
\end{code}
The following example illustrates their work on the initial segment {\tt [0..31]} of $\N$:
\begin{codex}
*InfPair> map (nAdicBij 2 3) [0..31]
[0,1,3,2,9,5,6,4,27,14,15,8,18,10,12,7,81,41,42,
 22,45,23,24,13,54,28,30,16,36,19,21,11]
*InfPair> map (nAdicBij 3 2) [0..31]
[0,1,3,2,7,5,6,15,11,4,13,31,14,23,9,10,27,63,
 12,29,47,30,19,21,22,55,127,8,25,59,26,95]
\end{codex}
It is easy to see that the following holds:
\begin{prop}
\begin{equation}
(\mathit{nAdicBij}~k~l) \circ (\mathit{nAdicBij}~l~k) \equiv \mathit{id}
\end{equation}
\end{prop}
As a side note, such bijections might have applications to cryptography, 
provided that a
method is devised to generate ``interesting'' pairs {\tt (k,l)}
defining the encoding.
%\end{comment}

We can derive {\tt Encoders} representing functions
between $\N$ and sequences
of natural numbers, parameterized by a (possibly infinite)
list of {\tt nAdicHead / nAdicTail} bases, by repeatedly applying
the $n$-adic head, tail and cons operation parameterized
by the (assumed infinite) sequence {\tt ks}:
\begin{code}
nAdicNats :: [N]->Encoder N
nAdicNats ks = Iso (nat2nAdicNats ks) (nAdicNats2nat ks)

nat2nAdicNats :: [N]->N->[N]
nat2nAdicNats _ 0 = []
nat2nAdicNats (k:ks) n | n>0 = 
  nAdicHead k n : nat2nAdicNats ks (nAdicTail k n)

nAdicNats2nat :: [N]->[N]->N
nAdicNats2nat _ [] = 0
nAdicNats2nat (k:ks) (x:xs) = 
  nAdicCons k (x,nAdicNats2nat ks xs)
\end{code}
For instance, the Encoder {\tt nat'} corresponds to {\tt ks}
defined as the infinite sequence
starting at {\tt 2}.
\begin{code}
nat' :: Encoder N
nat' = nAdicNats [2..]
\end{code}
The following examples illustrate the mechanism:
\begin{codex}
*InfPair> as nat' list [2,0,1,2]
1644
*InfPair> as list nat' it
[2,0,1,2]
*InfPair> map (as nat' nat) [0..15]
[0,1,2,3,4,7,6,5,8,19,14,15,12,13,10,9]
*InfPair> map (as nat' nat) [0..15]
[0,1,2,3,4,7,6,5,8,19,14,15,12,13,10,9]
\end{codex}
Note that
functions like {\tt as nat' nat} illustrate another general mechanism for defining permutations of 
$\N$.

\section{Pairing bijections derived from 
characteristic functions of subsets of $\N$} \label{char}

We start by connecting the bitstring representation of characteristic functions to our 
bijective data transformation framework (overviewed in the {\bf Appendix}).

\subsection{The bijection between lists and characteristic functions of sets}

The function {\tt list2bins} converts a sequence of natural numbers 
into a characteristic function of a subset of $\N$ represented as a string of binary
digits. The algorithm interprets
each element of the list as the number of {\tt 0} 
digits before the next {\tt 1} digit.
Note that infinite sequences are handled as well, resulting
in infinite bitstrings.
\begin{code}
list2bins :: [N]->[N]

list2bins [] = [0]
list2bins ns = f ns where
  f [] = []
  f (x:xs) = (repl x 0) ++ (1:f xs) where
    repl n a | n <= 0 = []
    repl n a = a:repl (pred n) a
\end{code}
The function {\tt bin2list} converts a characteristic function represented
as bitstrings back to a list of natural numbers.
\begin{code}
bins2list :: [N] -> [N]
bins2list xs = f xs 0 where
  f [] _ = []
  f (0:xs) k = f xs (k+1)
  f (1:xs) k = k : f xs 0
\end{code}
Together they provide the Encoder {\tt bins}, 
that we will use to connect characteristic functions
to various data types.
\begin{code} 
bins :: Encoder [N]
bins = Iso bins2list list2bins
\end{code}
The following examples (where the Haskell library function {\tt take} 
is used to restrict execution to an initial segment of an infinite list) 
illustrate their use:
\begin{codex}
*InfPair> list2bins [2,0,1,2]
[0,0,1,1,0,1,0,0,1]
*InfPair> bins2list it
[2,0,1,2]

*InfPair> take 20 (list2bins [0,2..])
[1,0,0,1,0,0,0,0,1,0,0,0,0,0,0,1,0,0,0,0]
*InfPair> bins2list it
[0,2,4,6]
\end{codex}
The following holds:
\begin{prop}
If $M$ is a subset of $\N$, the bijection
{\tt as bins set} returns the bitstring
associated to $M$ and its inverse is the bijection
{\tt as set bins}.
\end{prop}
\begin{proof}
Observe that the transformations are the composition
of bijections between bitstrings and lists and 
bijections between lists and sets.
\end{proof}
The following example illustrates this correspondence:
\begin{codex}
*InfPair> as bins set [0,2,4,5,7,8,9]
[1,0,1,0,1,1,0,1,1,1]
*InfPair> as set bins it
[0,2,4,5,7,8,9]
\end{codex}
Note that, for convenient use on finite sets, the functions
do not add the infinite stream of {\tt 0} digits indicating
its infinite stream of non-members, but we will add it as needed
when the semantics of the code requires it for
representing accurately operations on infinite sequences.
We will use the same convention through the paper.

\subsection{Splitting and merging bitstrings with a characteristic function}
Guided by the characteristic function of a subset of $\N$, represented as
a bitstring, the function {\tt bsplit} separates a (possibly infinite)
sequence of numbers into two lists: members and non-members.
\begin{code}
bsplit :: [N] -> [N] -> ([N], [N])
bsplit _ [] = ([],[])
bsplit [] (n:ns) = 
  error ("bspilt provides no guidance at: "++(show n))
bsplit (0:bs) (n:ns) = (xs,n:ys) where 
  (xs,ys) = bsplit bs ns 
bsplit (1:bs) (n:ns) = (n:xs,ys) where 
  (xs,ys) = bsplit bs ns 
\end{code}

Guided by the characteristic function of a subset of $\N$, represented as
a bitstring, the function {\tt bmerge} merges two lists of natural numbers
into one, by interpreting each {\tt 1} in the characteristic function
as a request to extract an element of the first list and each {\tt 0}
as a request to extract an element of the second list.  
\begin{code}
bmerge :: [N] -> ([N], [N]) -> [N]
bmerge _ ([],[]) = []
bmerge bs ([],[y]) = [y]
bmerge bs ([x],[]) = [x]
bmerge bs ([],ys) = bmerge bs ([0],ys)
bmerge bs (xs,[]) = bmerge bs (xs,[0])
bmerge (0:bs) (xs,y:ys) = y : bmerge bs (xs,ys)
bmerge (1:bs) (x:xs,ys) = x : bmerge bs (xs,ys)
\end{code}
The following examples (trimmed to finite lists) illustrate their use:
\begin{codex}
*InfPair> bsplit [0,1,0,1,0,1] [10,20,30,40,50,60]
([20,40,60],[10,30,50])
*InfPair> bmerge [0,1,0,1,0,1] it
[10,20,30,40,50,60]
\end{codex}

\begin{comment}
   genericUnpair set [0..] n => (n,0)
   
   to avoid degenerate cases - assumption:   
   "xs and its complement are both infinite"
   true for instance when
   xs is a strictly increasing sequence of natural numbers
\end{comment}

\subsection{Defining pairing bijections, generically}
We design a generic mechanism to
derive pairing functions by combining
the data type transformation operation {\tt as}
with the {\tt bsplit} and {\tt bmerge} functions
that apply a characteristic function encoded as a
list of bits.

\begin{code}
genericUnpair :: Encoder t -> t ->   N -> (N, N)
genericUnpair xEncoder xs n = (l,r) where 
  bs = as bins xEncoder xs
  ns = as bins nat n
  (ls,rs) = bsplit bs ns
  l = as nat bins ls
  r = as nat bins rs
\end{code}

\begin{code}
genericPair :: Encoder t -> t ->   (N, N) -> N
genericPair xEncoder xs (l,r) = n where
  bs = as bins xEncoder xs
  ls = as bins nat l
  rs = as bins nat r
  ns = bmerge bs (ls,rs)
  n = as nat bins ns
\end{code}

Let us observe first that for termination of
this functions depends on termination of 
the calls to {\tt bsplit} and {\tt bmerge},
as illustrated by the following examples:
\begin{codex}
*InfPair> genericPair bins (cycle [0]) (10,20)
^CInterrupted.
*InfPair> genericUnpair bins (cycle [1]) 42
(^CInterrupted.
\end{codex}
In this case, the characteristic functions
given by {\tt cycle [0]} or {\tt cycle [1]}
would trigger an infinite search for a non-existing
first {\tt 1} or {\tt 0} in {\tt bsplit} and
{\tt bmerge}.

Clearly, this suggests restrictions
on the acceptable characteristic functions.

We will now give sufficient conditions ensuring that
the functions {\tt genericUnpair} and {\tt genericPair}
terminate for any values of their last arguments.
Such restrictions, will enable them to define families of pairing
functions parameterized
by characteristic functions derived from
various data types.

\begin{df}
We call {\em bloc} of digits occurring in a characteristic
function any (finite or infinite) contiguous sequence of digits.
\end{df}
Note that an infinite bloc made entirely of $0$ (or $1$) digits can only
occur at the end of the sequence defining the characteristic function, 
i.e. only if it exists a number
$n$ such that the index of each member of the bloc 
is larger than $n$.

\begin{prop} \label{infset}
If $\{a_n\}_{n \in \N}$ is an infinite sequence of bits
containing only finite blocks of {\tt 0} and {\tt 1} digits,
{\tt genericPair bins} and {\tt genericUnPair bins}
define a family of pairing bijections
parameterized by $\{a_n\}_{n \in \N}$.
\end{prop}
\begin{proof}
Having an alternation of finite blocks of $1$s and $0$s, 
ensures that, when called from {\tt genericPair}
and {\tt genericUnPair}, 
the functions {\tt bmerge} and {\tt bsplit}
terminate.
\end{proof}

For instance, {\em Morton} codes \cite{lawder:2000} are
derived by using a stream of alternating 
{\tt 1} and {\tt 0} digits (provided by the Haskell library function {\tt cycle})
\begin{code}
bunpair2 = genericUnpair bins (cycle [1,0])
bpair2 = genericPair bins (cycle [1,0])
\end{code}
and working as follows:
\begin{codex}
*InfPair> map bunpair2 [0..10]
[(0,0),(1,0),(0,1),(1,1),(2,0),
 (3,0),(2,1),(3,1),(0,2),(1,2),(0,3)]
*InfPair> map bpair2 it
[0,1,2,3,4,5,6,7,8,9,10]
\end{codex}

\begin{prop} \label{infset}
If $\{a_n\}_{n \in \N}$ is an infinite sequence of non-decreasing
natural numbers, the functions {\tt genericPair set} and {\tt genericUnPair set}
define a family of pairing bijections
parameterized by $\{a_n\}_{n \in \N}$.
\end{prop}
\begin{proof}
Given that the sequence is non-decreasing, it represents canonically an infinite
set such that its complement is also infinite, represented as
a non-decreasing sequence. Therefore, the associated characteristic
function will have an alternation of finite blocks of {\tt 1} and {\tt 0} digits, 
inducing a pairing/unpairing bijection.
\end{proof}

The bijection {\tt bpair k} and its inverse {\tt bunpair k} are derived
from a {\tt set} representation (implicitly morphed into a characteristic
function).
\begin{code}
bpair k = genericPair set [0,k..]
bunpair k = genericUnpair set [0,k..] 
\end{code}
Note that for {\tt k = 2} we obtain exactly the bijections {\tt bpair2} and
{\tt bunpair2} derived previously, as illustrated by the following example:
\begin{codex}
*InfPair> map (bunpair 2) [0..10]
[(0,0),(1,0),(0,1),(1,1),(2,0),(3,0),
 (2,1),(3,1),(0,2),(1,2),(0,3)]
*InfPair> map (bpair 2) it
[0,1,2,3,4,5,6,7,8,9,10]
\end{codex}
We conclude with a similar result for lists:
\begin{prop}
If $\{a_n\}_{n \in \N}$ is an infinite sequence of 
natural numbers only containing finite blocks of 0s,
the functions {\tt genericPair list} and {\tt genericUnPair list}
define a family of pairing bijections
parameterized by $\{a_n\}_{n \in \N}$.
\end{prop}
\begin{proof}
It follows from Prop. \ref{infset} by observing that such sequences
are transformed into infinite sets represented as non-decreasing sequences.
\end{proof}
The {\bf Appendix} discusses a few more examples of such pairing functions
and visualizes a few space-filling curves associated to them.

\begin{prop}
There are $2^\N$ pairing functions defined using characteristic functions of sets 
of $\N$.
\end{prop}
\begin{proof}
Observe that a characteristic function corresponding to a subset of $\N$ 
containing an infinite bloc of {\tt 0} or {\tt 1} digits necessarily ends
with the bloc. Therefore, by erasing the bloc we can put such functions in a bijection
with a finite subset of $\N$. Given that there are only
a countable number of finite subsets of $\N$, the 
cardinality of the set of the 
remaining subsets' characteristic functions is $2^\N$.
\end{proof}

\section{Related Work} \label{rel}
Pairing functions have been used in work on decision problems as early as
\cite{pepis,kalmar1,robinson50,robinson55,robinson68a,robinsons68b}.
There are about {\tt 19200} Google documents
referring to the original ``Cantor pairing function''
among which we mention the surprising
result that, together with the successor function
it defines a decidable subset of arithmetic
\cite{DBLP:journals/tcs/CegielskiR01}.
An extensive study of various pairing functions and their 
computational properties is presented in 
\cite{ceg99,DBLP:conf/ipps/Rosenberg02a}.
They are also
related to 2D-space filling curves
(Z-order, Gray-code and Hilbert curves) 
\cite{lawder99,lawder:2000,Lawder:2001,faloutsos:2001}.
Such curves are obtained by connecting pairs of coordinates
corresponding to successive natural numbers (obtained by applying 
unpairing operations).
They have applications to spatial and multi-dimensional database indexing
\cite{lawder99,lawder:2000,Lawder:2001,faloutsos:2001} and
symbolic arbitrary length arithmetic computations \cite{sac12}.
Note also that {\tt bpair 2} and {\tt bunpair 2} are the same as the
functions defined in \cite{pigeon} and also known as Morton-codes,
with uses in indexing of spatial databases \cite{lawder99}.

\begin{comment}
Finally, the ``once you have seen them, obvious'' 
{\tt list2set / set2list, list2mset/mset2list} bijections 
given in the {\bf Appendix} are borrowed from
\cite{calc09fiso}, but not unlikely to be common
knowledge of people working in 
combinatorics or recursion theory.
These simple bijections between lists and sets
of natural numbers 
show the unexpected usefulness of
the framework supporting bijective data type
transformations \cite{calc09fiso}.
\end{comment}

\section{Conclusion} \label{concl}

We have described mechanisms for generating 
countable and uncountable families of pairing / unpairing
bijections. The mechanism involving
$n$-adic valuations is definitely novel,
and we have high confidence (despite of their
obviousness) that the characteristic function-based
mechanisms are novel as well, at least in terms
of their connections to list, set or multiset representations
provided by the implicit use of our bijective
data transformation framework \cite{calc09fiso}.

Given the space constraints, we
have not explored the natural
extensions to more general tupling / untupling
bijections (defined between $\N^k$ and $\N$) as well
as bijections between finite lists, sets and
multisets that can be derived quite easily, using
the data transformation framework given in the
{\bf Appendix}. For the same reasons we have
not discussed specific applications of these
families of pairing functions, but we foresee
interesting connections with possible
cryptographic uses (e.g ``one time pads''
generated through intricate combinations of
members of these families). 

The ability
to associate such pairing functions
to arbitrary characteristic functions
as well as to their equivalent set, multiset, list
representations provides convenient
tools for inventing and customizing pairing / unpairing
bijections, as well as the related
tupling / untupling bijections and those defined
between natural numbers and sequences,
sets and multisets of natural numbers.

We hope that our adoption of the non-strict
functional language Haskell (freely available
from \url{haskell.org}), as a complement to conventional
mathematical notation, enables the empirically curious
reader to instantly validate our claims and encourage
her/him to independently explore their premises and
their consequences.

\bibliographystyle{plain}
\bibliography{INCLUDES/theory,go/tarau,INCLUDES/proglang,INCLUDES/biblio,INCLUDES/syn}

%\newpage
\section*{Appendix}

\subsection*{An Embedded Data Transformation Language}

We will describe briefly the embedded data transformation 
language used
in this paper as a set of operations on a groupoid of isomorphisms. 
We refer to (\cite{calc09fiso,arxiv:fISO}) for details.

\subsubsection*{The Groupoid of Isomorphisms}

We implement an isomorphism between two objects X and Y as a 
Haskell data type encapsulating a bijection $f$ and its inverse $g$. 

%\begin{comment}

\begindc{\commdiag}[5]
\obj(14,11){$X$}
\obj(39,11){$Y$}
\mor(14,12)(39,12){$f=g^{-1}$}
\mor(39,10)(14,10){$g=f^{-1}$}
\enddc

We will call the {\em from} function the first component (a {\em
section} in category theory parlance) and
the {\em to} function the second component (a {\em retraction}) defining
the isomorphism.
The isomorphisms are naturally organized as a {\em groupoid}.
%\end{comment}
\begin{code}
data Iso a b = Iso (a->b) (b->a)

from (Iso f _) = f

to (Iso _ g) = g

compose :: Iso a b -> Iso b c -> Iso a c
compose (Iso f g) (Iso f' g') = Iso (f' . f) (g . g')

itself = Iso id id

invert (Iso f g) = Iso g f
\end{code}
Assuming that for any pair of type {\tt Iso a b},  $f \circ g = id_b$ and $g
\circ f=id_a$, we can now formulate {\em laws} about these isomorphisms.

\vskip 5mm

\noindent {\em
The data type {\tt Iso} has a groupoid structure, i.e. the {\em compose}
operation, when defined, is associative, {\em itself} acts 
as an identity element
and {\em invert} computes the inverse of an isomorphism.
}
\subsubsection*{The Hub: Sequences of Natural Numbers}
To avoid defining $\frac{n(n-1)}{2}$ isomorphisms between $n$ objects,
we choose a {\em Hub} object to/from which we will actually
implement isomorphisms.

Choosing a {\em Hub} object is somewhat arbitrary, but it makes sense to
pick a representation that is relatively easy convertible to various
others and scalable to
accommodate large objects up to the runtime system's 
actual memory limits.

We will choose as our {\tt Hub} object {\em sequences of natural
numbers}.
We will represent them as lists i.e. their Haskell type is {\tt [N]}.
\begin{code}
type N = Integer
type Hub = [N]
\end{code}
We can now define an {\tt Encoder} as an isomorphism
connecting an object to {\em Hub} 
\begin{code}
type Encoder a = Iso a Hub
\end{code}
together with the combinator ``{\tt as}'',
providing an {\em embedded transformation language} for routing
isomorphisms through two {\tt Encoders}.
\begin{code}  
as :: Encoder a -> Encoder b -> b -> a
as that this x = g x where Iso _ g = compose that (invert this)
\end{code}
The combinator ``{\tt as}'' adds a convenient syntax
such that converters between {\tt A} and {\tt B} can be designed as:
\begin{codex}
a2b x = as B A x
b2a x = as A B x
\end{codex}

%\begin{comment}
\vskip 0.30cm

\begindc{\commdiag}[5]
\obj(26,0){$Hub$}
\obj(14,11){$A$}
\obj(39,11){$B$}

\mor(26,0)(39,10){$b$}
\mor(26,0)(14,10){$a^{-1}$}
\mor(39,10)(26,0){$b^{-1}$}
\mor(14,10)(26,0){$a$}
\mor(14,12)(39,12){$a2b=as~B~A$}
\mor(39,10)(14,10){$b2a=as~A~B$}
\enddc
\vskip 0.30cm
%\end{comment}
\begin{comment}
A particularly useful combinator that
transports binary operations from an Encoder to another, {\tt
borrow\_from}, can be defined as follows:
\begin{code}
borrow_from :: Encoder a -> (a -> a -> a) -> 
               Encoder b -> (b -> b -> b)
borrow_from lender op borrower x y = as borrower lender
   (op (as lender borrower x) (as lender borrower y))
\end{code}
\end{comment}
\noindent Given that {\tt [N]} has been chosen as the root, we will define our
 sequence data type {\em list} simply as the identity isomorphism 
on sequences in {\tt [N]}.
\begin{code}  
list :: Encoder [N]
list = itself
\end{code} 
The {\tt Encoder} {\tt mset} for multisets of natural numbers is defined as:
\begin{code}
mset :: Encoder [N]
mset = Iso mset2list list2mset

mset2list, list2mset :: [N]->[N]
mset2list xs = zipWith (-) (xs) (0:xs)
list2mset ns = tail (scanl (+) 0 ns) 
\end{code}
The {\tt Encoder} {\tt set} for sets of natural numbers is defined as:
\begin{code}
set :: Encoder [N]
set = Iso set2list list2set

set2list, list2set :: [N]->[N]
list2set = (map pred) . list2mset . (map succ)
set2list = (map pred) . mset2list . (map succ) 
\end{code}
Note that these converters between lists, multisets and sets make no assumption about
finiteness of their arguments and therefore they
can used in a non-strict language like Haskell on infinite objects as well.

\subsection*{Examples of pairing functions derived from characteristic functions}

The function {\tt syracuse} is used in an equivalent formulation of the Collatz conjecture. Interestingly, it can be computed using the {\tt nAdicTail} which results after
dividing a number $n$ with $\mu_2(n)$. Note that we derive our pairing function directly
from the {\tt list} representation of the range of this function as {\tt genericPair} and 
{\tt genericUnpair} implicitly construct the associated characteristic function.
\begin{code}
syracuse :: N->N
syracuse n = nAdicTail 2 (6*n+4)

nsyr 0 = [0]
nsyr n = n : nsyr (syracuse n)
\end{code}

\begin{code}
syrnats = map syracuse [0..]

syrpair = genericPair list syrnats
syrunpair = genericUnpair list syrnats 
\end{code}

Figures \ref{bunpair2} and \ref{bunpair3} 
show the ``{\em Z-order}'' (Morton code) 
path connecting successive values 
in the range of the function {\tt bunpair 2} and {\tt bunpair 3}. 
Figures \ref{syrUnpair} and
\ref{piUnpair} show the path connecting the values in the range of unpairing functions
associated, respectively to the Syracuse function and the binary digits of $\pi$.
Interestingly, at a first glance, some regular patterns emerge even in the case of
such notoriously irregular characteristic functions.
\FIG{bunpair2}{Path connecting values of {\tt bunpair 2}}{0.30}{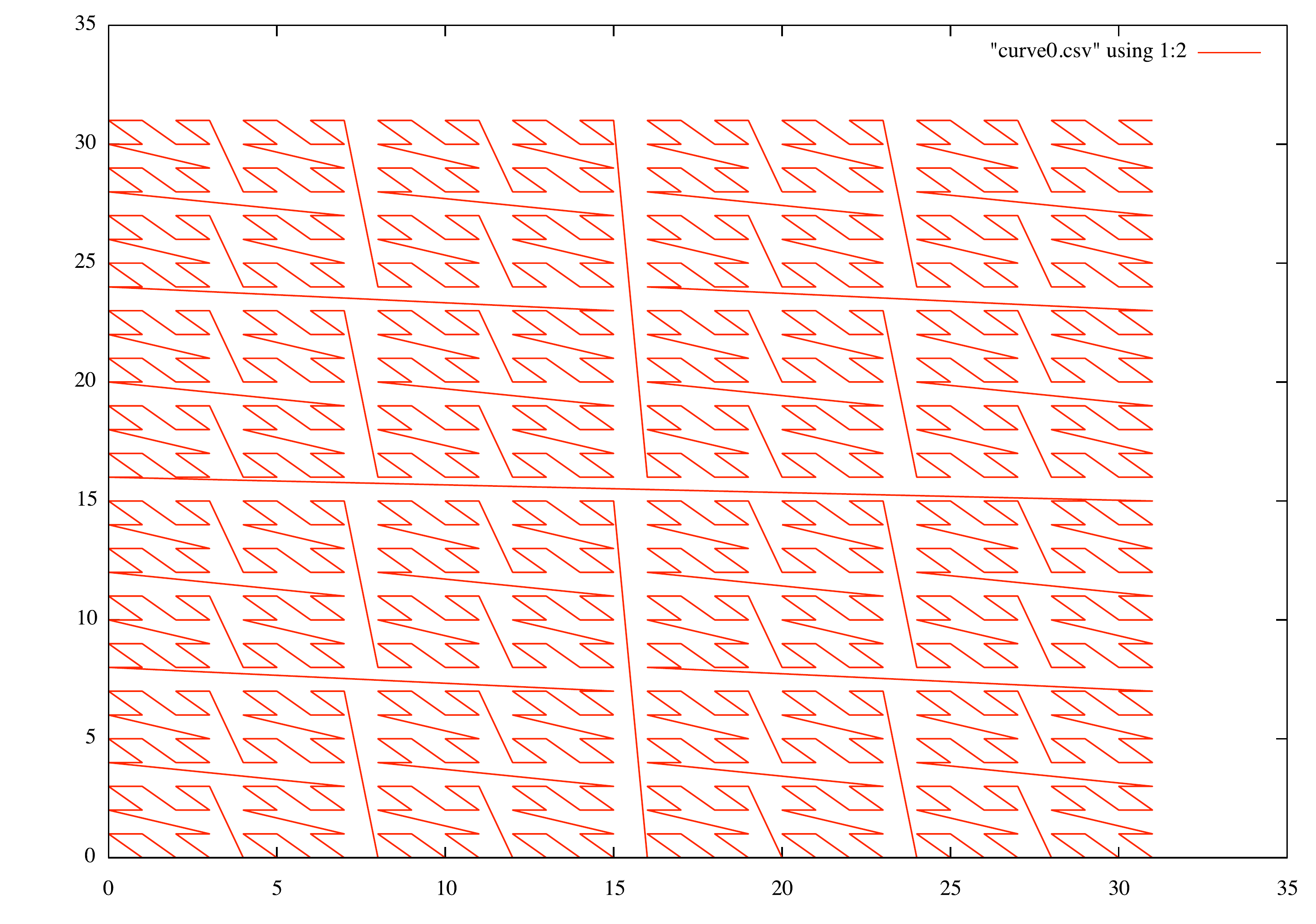}
\FIG{bunpair3}{Path connecting values of bunpair 3}{0.30}{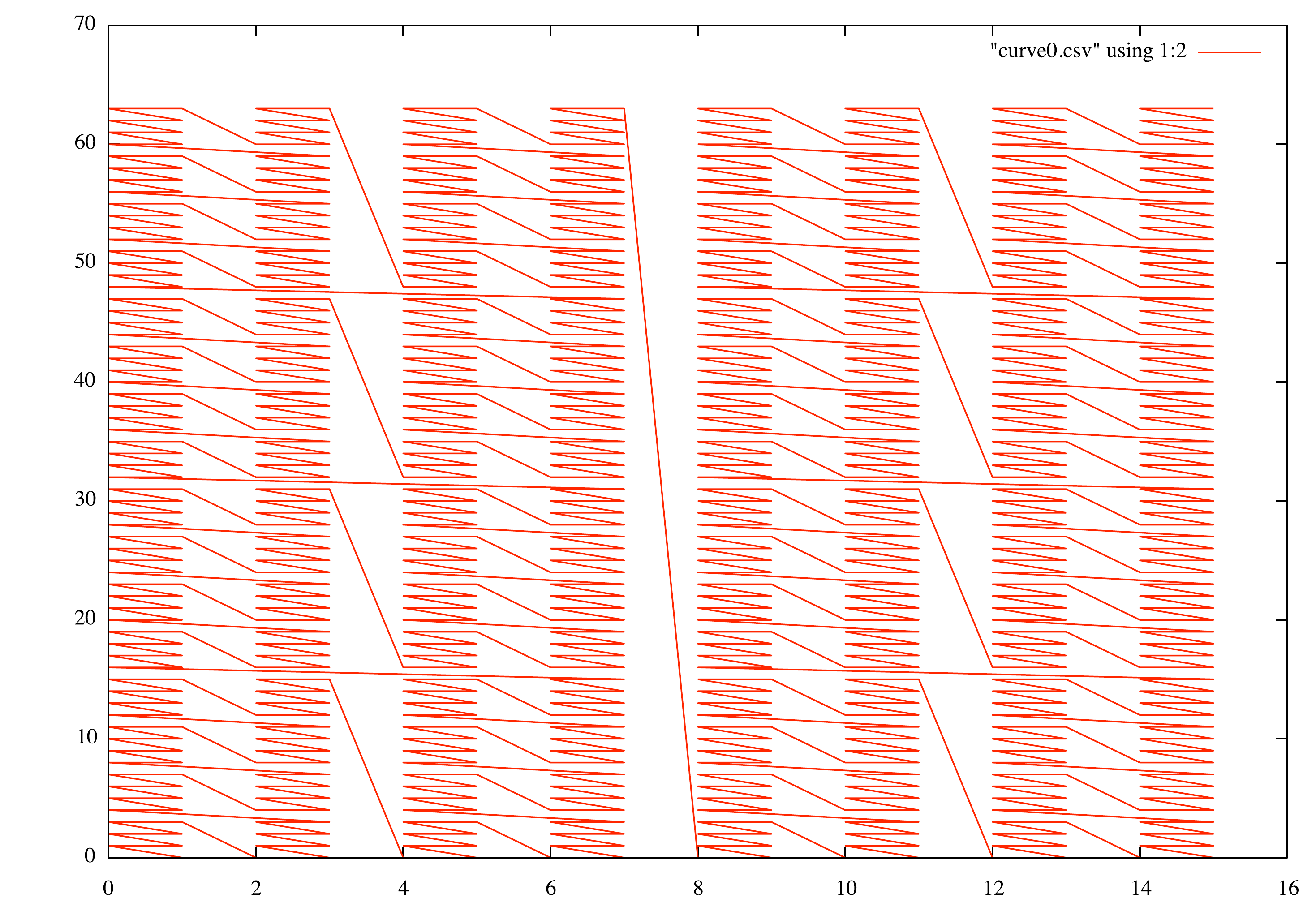}
\FIG{syrUnpair}{Path connecting values of an unpairing bijection based on the Syracuse function}{0.30}{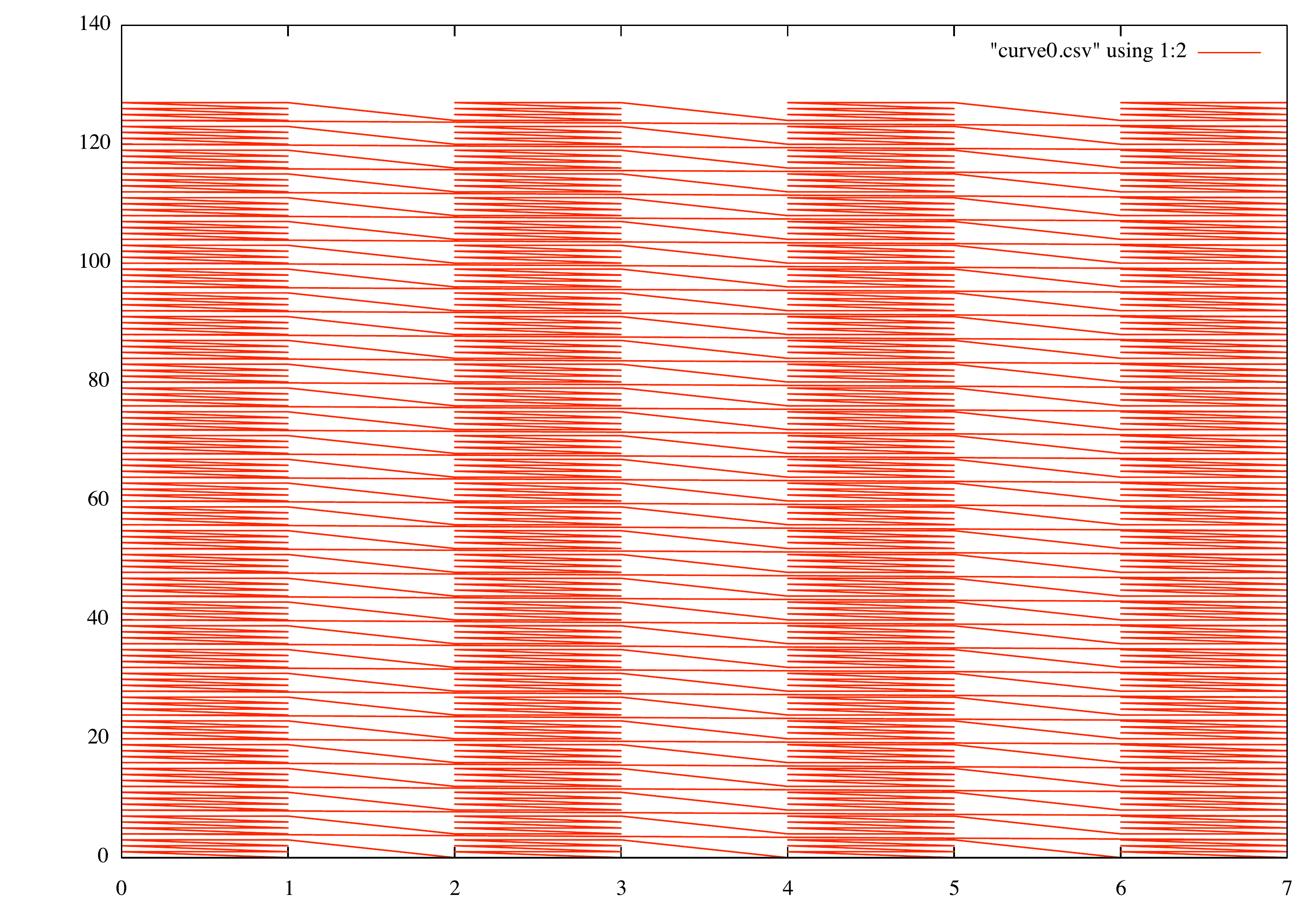}
\FIG{piUnpair}{Path connecting values of an unpairing bijection based on binary digits of $\pi$}{0.30}{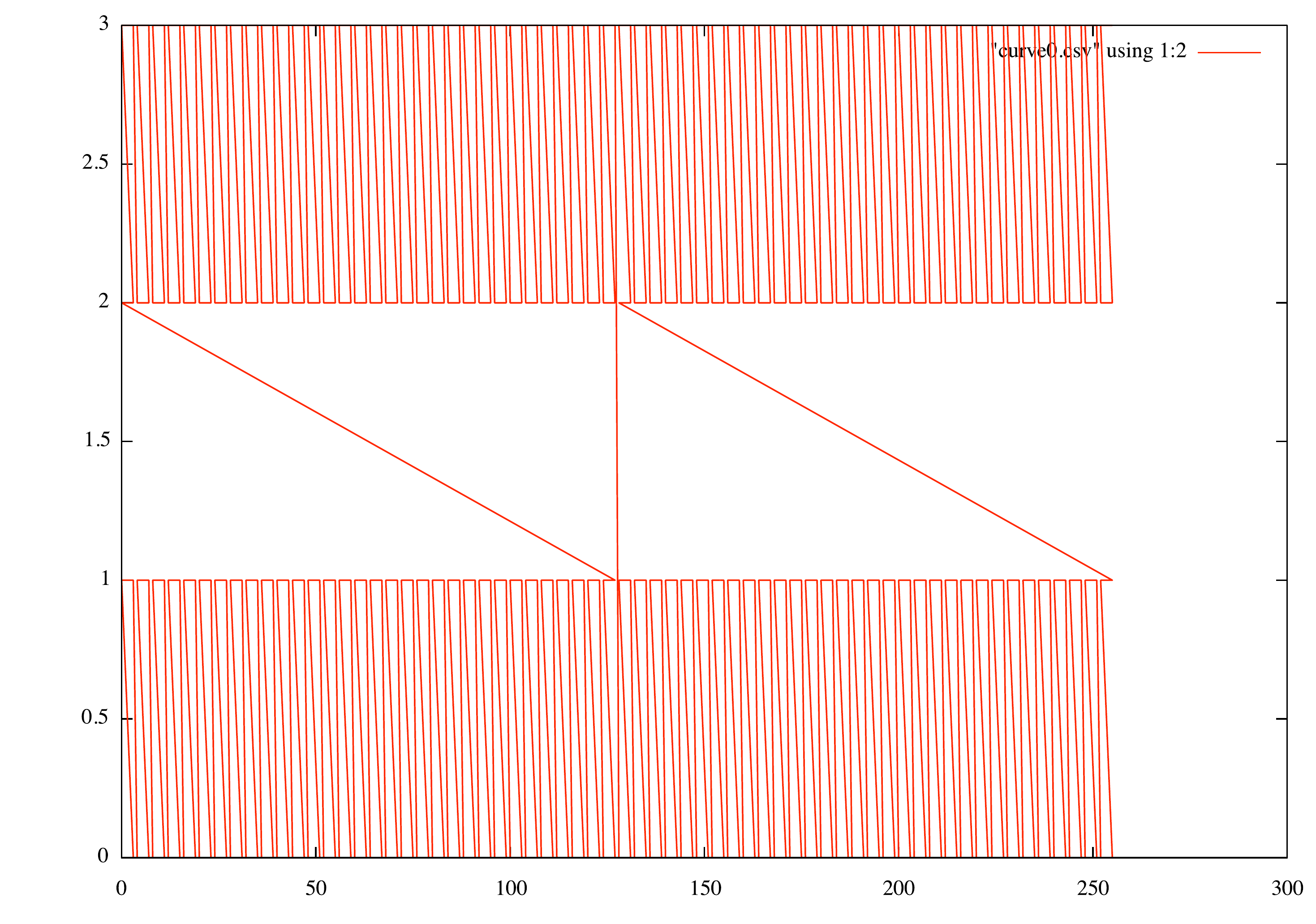}

%\VSFIGS{dd}{2012}{as 1}{as 2}{bunpair2.pdf}{bunpair3.pdf}{0.40}
%\VSFIGS{dd}{2012}{as 1}{as 2}{piUnpair.pdf}{syrUnpair.pdf}{0.40}

% -- end --

\begin{codeh}
sqpair = genericPair set (map (^2) [0..])
squnpair = genericUnpair set (map (^2) [0..]) 
\end{codeh}

\begin{codeh}
npair = genericPair list [0..]
nunpair = genericUnpair list [0..] 
\end{codeh}

\begin{codeh}
bnats = concatMap (as bins nat) [0..]

bnatpair = genericPair bins bnats
bnatunpair = genericUnpair bins bnats 
\end{codeh}

\begin{codeh}
powunpair = genericUnpair set (map (2^) [0..])
powpair = genericPair set (map (2^) [0..])
\end{codeh}

\begin{codeh}
{-
-- add "import Pi" to the top of this file
-- using PI as a source of a bitstream
bin_pi = as bins nat (machin_pi (2^12))

pi_pair = genericPair bins bin_pi
pi_unpair = genericUnpair bins bin_pi 
-}
\end{codeh}

% information lost by unpairing

\begin{codeh}
infs u n = (bsize n) - s where
  (a,b) = u n
  s = (bsize a)+(bsize b)

bsize 0 = 0
bsize n | n>0 = 1 + (bsize (drop2digit n))  where
  drop2digit n = (shiftR n 1)+(1 .&. n)-1

xunp u x n = u (n `xor` x)

xp p x (a,b) = (p (a,b)) `xor` x

xtest p u (x,n) = xp p x (xunp u x n)

xplot = pplot (xunp (bunpair2) 1) 7  

xtest1 = map (xunp (bunpair2) 7) [0..15]
\end{codeh}

\end{document}